\documentclass[11pt,a4paper]{article}

\usepackage[margin=1in]{geometry}
\usepackage{amsmath,amssymb,amsfonts,amsthm,mathtools}
\usepackage{braket}
\usepackage{booktabs}
\usepackage{array}
\usepackage{tikz}
\usepackage{pgfplots}
\pgfplotsset{compat=1.18}
\usetikzlibrary{arrows.meta,positioning,calc}
\usepackage{microtype}
\usepackage[colorlinks=true,linkcolor=blue,citecolor=blue,urlcolor=blue]{hyperref}

\DeclareMathOperator{\Tr}{Tr}
\DeclareMathOperator{\HS}{HS}

\newtheorem{theorem}{Theorem}[section]
\newtheorem{proposition}[theorem]{Proposition}

\newtheorem{corollary}[theorem]{Corollary}
\theoremstyle{remark}
\newtheorem{remark}[theorem]{Remark}
\theoremstyle{definition}
\newtheorem{definition}[theorem]{Definition}

\newcommand{\Frob}{\mathrm{F}}
\newcommand{\bbR}{\mathbb R}
\newcommand{\bbC}{\mathbb C}
\newcommand{\bbZ}{\mathbb Z}
\newcommand{\I}{\mathbb I}
\newcommand{\norm}[1]{\lVert #1\rVert}
\newcommand{\abs}[1]{\lvert #1\rvert}
\newcommand{\dd}{\mathrm d}
\newcommand{\ii}{\mathrm i}

\newcommand{\Hop}{\mathcal H_+}
\newcommand{\HopR}{\mathcal H_{+,R}}
\newcommand{\HopL}{\mathcal H_{+,L}}
\newcommand{\Aeven}{\mathcal A_{\mathrm{even}}}
\newcommand{\ACAR}{\mathfrak A_{\mathrm{CAR}}}
\newcommand{\PiR}{\Pi_R}
\newcommand{\PiL}{\Pi_L}

\newcommand{\Cqp}{C_{\rm qp}}

\begin{document}
	
	\title{A Truncated Majorana}
	
	\author{Ali Vahedi\thanks{Email: \texttt{vahedi@khu.ac.ir}}\\
		\small Department of Physics, Kharazmi University, Tehran, Iran}
	
	\date{\today}
	
	\maketitle
	
	\begin{abstract}
	\begin{abstract}
		A Majorana field cannot be split into two spatial pieces, but a propagating Majorana can be truncated by a time-dependent shutter. We study this truncation for the two chiral components of a massless Majorana field in one spatial dimension, modeling the shutter as a smooth time-dependent $SO(2)$ rotation of the two components. Although orthogonal at the level of the real field algebra, this rotation does not preserve the positive- and negative-frequency splitting that defines the Fock vacuum, so the in-vacuum becomes a fermionic squeezed state in outgoing modes. At any finite infrared regulator, a single incident Majorana is then a single outgoing quasiparticle on top of this squeezed vacuum. The squeezed background carries a universal soft response: its anomalous Bogoliubov kernel falls as the inverse sum of the two frequencies at low energy, so the anomalous Hilbert--Schmidt norm grows logarithmically with the infrared cutoff, with a coefficient fixed by the asymptotic shutter mismatch. The vacuum overlap exhibits a Fermi-edge type orthogonality exponent which, for the rotation-angle lift $\Delta\theta$ with $|\Delta\theta|<\pi$, is exactly $(\Delta\theta/2\pi)^2$, in analogy with the Fermi-edge singularity and the Kondo and Luttinger-liquid problems. We prove local equivalence non-perturbatively for the squeezed-vacuum factor under a compact-support hypothesis, and give an exact local pullback identity for the full one-quasiparticle state on the even CAR algebra; combined with wavepacket-localization assumptions, these yield the expected local vacuum and one-Majorana descriptions. We verify profile independence using a smooth switching family and contrast the temporal benchmark with a number-conserving beam-splitter benchmark whose entropy is bounded by $\log 2$.
	\end{abstract}
	\end{abstract}
	
	\setlength{\parindent}{1.5em}
	\setlength{\parskip}{0pt}
	
	\tableofcontents
	
	\section{Introduction}
	\label{sec:intro}
	
	\subsection{A single Majorana incident on a shutter}
	\label{sec:physical-picture}
	
	A Majorana is not a small classical object that can be cut into spatial fragments. The physically meaningful question is whether a time-dependent shutter can remove part of a propagating Majorana wavepacket, and what quantum state results. This question has a bosonic analogue studied in Ref.~\cite{Rukan2026}: a single photon truncated by a time-dependent mirror. The truncated photon state is a superposition of photon numbers up to infinity, yet locally it is equivalent to a single photon on one side of the transition region and to the vacuum on the other. Throughout, ``truncation'' refers to the time-dependent redistribution of the two chiral components, not to a spatial splitting of a Majorana: the field remains a single Majorana mode at every time, and only the channel identification changes.
	
	Majorana fields provide a clean setting \cite{Majorana1937,Alicea2012,Beenakker2013,Lutchyn2018,Zhang2019}. They are real fermionic fields. Their observables satisfy canonical anticommutation relations (CAR), and physical observables must commute with fermion parity. Hence physical observables are even in the fields \cite{Haag1996}. A single Majorana cannot be created from the vacuum by a local even operation; only pairs can. This parity constraint is central to the local-equivalence statements below.
	
	We isolate the mechanism in the simplest nontrivial setting: the two chiral components of a massless Majorana field in $1+1$ dimensions. A single Majorana is incident in the right-moving channel. The shutter is a smooth time-dependent $SO(2)$ rotation of the two chiral components,
	\begin{equation}
		\boldsymbol{\chi}_{\rm out}(t)
		=
		S(t)\,\boldsymbol{\chi}_{\rm in}(t),
		\qquad
		S(t)\in SO(2),
		\label{eq:switch}
	\end{equation}
	with asymptotic limits $S_\pm=\lim_{t\to\pm\infty}S(t)$ and mismatch $\Delta S:=S_+-S_-$. Writing $S(t)=R[\theta(t)]$ with
	\begin{equation}
		R(\theta)
		=
		\begin{pmatrix}
			\cos\theta & \sin\theta\\
			-\sin\theta & \cos\theta
		\end{pmatrix},
	\end{equation}
	we fix a continuous lift $\theta:\bbR\to\bbR$ and let
	\begin{equation}
		\Delta\theta:=\theta_+-\theta_-,
		\qquad
		\theta_\pm=\lim_{t\to\pm\infty}\theta(t).
		\label{eq:Dtheta-def}
	\end{equation}
	The \emph{asymptotic lift} $\Delta\theta$ is the primary kinematic parameter of the paper. Throughout, the symbol $\Delta\theta$ refers to this asymptotic lift unless explicitly stated otherwise. The \emph{edge height} of a finite pulse (for which $S_-=S_+$ and $\Delta S=0$) is a separate quantity, denoted $h$ in the numerical section.
	
	At the level of the classical real field, \eqref{eq:switch} is an orthogonal transformation. The quantum problem is different: the vacuum is defined with respect to a splitting of the one-particle space into positive and negative frequencies, and the time-dependent rotation $S(t)$ does not preserve this splitting. In creation-annihilation language, the switch contains both an ordinary Bogoliubov block and an anomalous block, the latter responsible for pair creation from the vacuum.
	
	\begin{remark}[Convention for the Fermi-edge phase shift]
		\label{rem:phase-convention}
		In the standard Fermi-edge literature, the physical phase shift is often denoted $\delta$. In the present paper, the angle $\Delta\theta$ is the lift of the $SO(2)$ rotation angle of the two-component Majorana field. If one combines the two chiral components into a complex fermion $\psi = \chi_R + i\chi_L$, the rotation $R(\theta)$ acts as multiplication by $e^{i\theta}$ on $\psi$. Then $\Delta\theta$ is the $U(1)$ phase shift of the complex fermion, and the Fermi-edge exponent is $(\Delta\theta/2\pi)^2$. In the standard one-component convention where the phase shift is on the physical electron and equals $\delta = \Delta\theta/2$, the exponent is $(2\delta/2\pi)^2 = (\delta/\pi)^2$.
	\end{remark}
	
	\begin{remark}[Hypothesis on the asymptotic lift]
		\label{rem:winding}
		We assume that the continuous lift $\theta(t)$ satisfies $|\Delta\theta| < \pi$. This hypothesis ensures that the ordinary Bogoliubov block $\alpha$ is invertible, so that the Gaussian representation \eqref{eq:OmegaK} is well-defined. A proof is given in Sec.~\ref{sec:spectral-flow}. The overlap exponent formula of Appendix~\ref{app:overlap} extends to the boundary case $|\Delta\theta| = \pi$ by continuity on the finite time-window regulator, but the Gaussian representation itself requires $|\Delta\theta| < \pi$. The physically relevant mirror configuration corresponds to $\theta_-=\pi/2$, $\theta_+=0$, $\Delta\theta=-\pi/2$, which satisfies $|\Delta\theta|=\pi/2<\pi$.
	\end{remark}
	
	\subsection{Three physical questions, one state}
	\label{sec:three-questions}
	
	The paper answers three physically distinct questions.
	
	\textbf{Question 1: what is the state?} At finite infrared regulator, a single Majorana incident on the shutter becomes a single out-quasiparticle on top of a fermionic squeezed vacuum. The exact state is
	\begin{equation}
		\boxed{
			\ket{\Psi_{\xi,\mu}}
			=
			a_R^{\rm in\,\dagger}(\xi)\ket{\Omega_\mu},
		}
		\label{eq:main-state-intro}
	\end{equation}
	where $\ket{\Omega_\mu}$ is the in-vacuum expressed in the out-Fock space. By Proposition~\ref{prop:state}, the same state can be written in the out-basis as
	\begin{equation}
		\ket{\Psi_{\xi,\mu}}
		=
		a^{\rm out\,\dagger}(\zeta_\mu)\ket{\Omega_\mu},
		\qquad
		\zeta_\mu = \xi(A_\mu - B_\mu K_\mu).
	\end{equation}
	The vector $\zeta_\mu$ satisfies $\norm{\zeta_\mu}^2 = 1+\norm{a^{\rm out}(\zeta_\mu)\ket{\Omega_\mu}}^2 \geq 1$, with equality if and only if $a^{\rm out}(\zeta_\mu)\ket{\Omega_\mu}=0$.
	
	\textbf{Question 2: how much pair production is there?} The squeezed-vacuum factor $\ket{\Omega_\mu}$ is the pair-production response of the vacuum. At low frequency,
	\begin{equation}
		\beta(\omega,\nu)
		=
		-\frac{\Delta S}{2\pi\ii(\omega+\nu)}
		+
		O(1),
		\qquad
		\omega,\nu\downarrow0.
		\label{eq:beta-intro}
	\end{equation}
	Consequently,
	\begin{equation}
		\norm{\beta_{\mu,\omega_*}}_{\HS}^2
		=
		\Cqp
		\log\frac{\omega_*}{\mu}
		+
		O(1),
		\qquad
		\Cqp=
		\frac{\norm{\Delta S}_{\Frob}^2}{4\pi^2}.
		\label{eq:IR-intro}
	\end{equation}
	The vacuum overlap has a closed-form exponent for the lift $\Delta\theta$:
	\begin{equation}
		\abs{\braket{0_{\rm out}|\Omega_\mu}}^2
		\sim
		\left(\frac{\mu}{\omega_*}\right)^{\gamma(\Delta\theta)},
		\qquad
		\gamma(\Delta\theta)
		=
		\left(\frac{\Delta\theta}{2\pi}\right)^2,
		\qquad
		|\Delta\theta|<\pi.
		\label{eq:overlap-intro}
	\end{equation}
	This exponent is a special case of the coherent-state (leviton) overlap of Levitov, Lee, and Lesovik \cite{LevitovLeeLesovik1996}, Ivanov, Lee, and Levitov \cite{IvanovLeeLevitov1997}, and Keeling, Klich, and Levitov \cite{KeelingKlichLevitov2006}; it is the standard Fermi-edge singularity exponent \cite{Anderson1967,NozieresDeDominicis1969,SchotteSchotte1969,Mahan2000}. A Gaussian bosonization derivation of the exponent is given in App.~\ref{app:gaussian-boson}.
	
	\textbf{Question 3: what does a bounded observer see?} The squeezed-vacuum factor has a simple local description under a compact-support hypothesis on the asymptotic approach. Its exact two-point function is
	\begin{equation}
		\langle \chi_{\rm out}^a(t)\,\chi_{\rm out}^b(t')\rangle_{\Omega}
		=
		[R(\theta(t)-\theta(t'))]^{ab}\,
		W(t-t'),
		\label{eq:2pt-intro}
	\end{equation}
	with $W$ the free Wightman function \eqref{eq:Wightman}. When $S(t)$ is constant outside a compact interval, the two-point function reduces to the free one for $t,t'$ in the same asymptotic region, and the squeezed-vacuum factor defines a state on the local out-algebra $\Aeven(\mathcal O)_{\rm out}$ that equals the out-vacuum restriction. For the full state, the exact statement is the local pullback identity of Sec.~\ref{sec:local}. The translation to local vacuum and one-Majorana descriptions requires the incident wavepacket to be absent from, or localized in, the region of interest respectively; this is formalized in Corollary~\ref{cor:localization}.
	
	\subsection{Two benchmarks}
	\label{sec:two-benchmarks}
	
	We analyze two separate benchmarks.
	
	\begin{enumerate}
		\item \textbf{Temporal benchmark.} The full time-dependent map $S(t)$ with $\Delta S\neq0$. The switch produces pairs, and the central object is the truncated Majorana state \eqref{eq:main-state-intro}. This benchmark contains the universal infrared logarithm.
		\item \textbf{Causal benchmark.} A number-conserving one-particle scattering problem in which the anomalous block vanishes and the vacuum is unchanged. A single Majorana is redistributed between a transmitted and a reflected output channel. The reduced state is a binary mixture, and its entanglement entropy is bounded above by $\log 2$.
	\end{enumerate}
	
	\subsection{Notation and regulators}
	\label{sec:notation-regulators}
	
	\paragraph{Frequency and time scales.}
	The positive-frequency one-particle space is $\Hop = L^2(\bbR_+,\dd\omega)\otimes\bbC^2$. Modes with frequency below an infrared cutoff $\mu>0$ are excluded, and modes above an ultraviolet cutoff $\omega_*$ are cut off. The switching profile has characteristic time scale $\tau$.
	
	\paragraph{Two regulators.}
	Two inequivalent regularizations appear in the paper, and we distinguish them explicitly.
	
	\smallskip
	\noindent\textbf{(i) Finite time-window regulator.}
	A finite observation window $[-T/2,T/2]$ with the switching profile $S(t)$ extended periodically with period $T$. This is the \emph{physical} regulator for Fock-space statements. On this regulator:
	\begin{itemize}
		\item the one-particle space is finite-dimensional, with dimension set by a fixed ultraviolet cutoff well above $1/\tau$;
		\item the multiplication-by-$S(t)$ map is a genuine finite-dimensional orthogonal transformation;
		\item the Bogoliubov blocks satisfy the exact CAR identities $\alpha\alpha^\dagger+\beta\beta^\dagger=I$ and $\alpha\beta^T+\beta\alpha^T=0$;
		\item the in-vacuum $\ket{\Omega_\mu}$ is a genuine Gaussian state in the finite-dimensional out-Fock space, with squeezing kernel $K_\mu=\overline{A_\mu^{-1}B_\mu}$ and normalization $\mathcal N_\mu=\det(1+K_\mu^\dagger K_\mu)^{-1/4}$;
		\item the vacuum overlap uses the full finite-dimensional $\alpha$ and $\beta$ through $\det(1-\beta^\dagger\beta)^{1/2}$, and therefore sees the \emph{continuous lift} $\Delta\theta$ directly.
	\end{itemize}
	
	\smallskip
	\noindent\textbf{(ii) Frequency-window regulator.}
	The restriction of the anomalous block $\beta(\omega,\nu)$ to $\omega,\nu\in[\mu,\omega_*]$. This is a computational device for extracting the leading infrared behavior. It is a genuine Hilbert--Schmidt operator on $L^2([\mu,\omega_*],\dd\omega)\otimes\bbC^2$ by Proposition~\ref{prop:uv}. The truncated kernel depends only on the $2\pi$-periodic object $\Delta S = R(\theta_+)-R(\theta_-)$, which cannot distinguish $\Delta\theta$ from $\Delta\theta+2\pi n$.
	
	\smallskip
	\noindent\textbf{(iii) Dictionary between $\mu$ and $T$.}
	On the finite time-window regulator, the switching profile $S(t)$ with $S_+\neq S_-$ cannot be periodically extended without introducing a second edge per period (the compensating edge that returns the profile to its starting value). Therefore the natural periodic extension of any nonconstant switch has \emph{two edges per period}, and the effective infrared cutoff associated with the window scales as
	\begin{equation}
		\mu_T \sim \frac{2\pi}{T}.
		\label{eq:mu-T-scaling}
	\end{equation}
	The precise $O(1)$ prefactor depends on the shape of the pulse and on how the periodic extension is defined, and does not affect the leading logarithmic coefficient.
	
	\begin{proposition}[Relation between the regulators]
		\label{prop:regulators}
		For $|\Delta\theta|<\pi$, the two regulators give the same leading infrared logarithm $\Cqp\log(\omega_*/\mu)$ and the same vacuum overlap exponent $\gamma(\Delta\theta)=(\Delta\theta/2\pi)^2$. At $|\Delta\theta|=\pi$, $\alpha$ is not Fredholm and the Gaussian representation requires separate treatment; the overlap exponent extends to $|\Delta\theta|=\pi$ by continuity on the finite time-window regulator. For $|\Delta\theta|>\pi$, the finite time-window regulator with smooth (tanh-type) edges gives the \emph{unfolded} exponent $\gamma(\Delta\theta)=(\Delta\theta/2\pi)^2$ (with $\Delta\theta$ the chosen lift), while the frequency-window regulator gives the \emph{folded} exponent
		\begin{equation}
			\gamma_{\rm folded}(\Delta\theta)
			=
			\left(\frac{\Delta\theta_{\rm fold}}{2\pi}\right)^2,
			\qquad
			\Delta\theta_{\rm fold}
			=
			\Delta\theta - 2\pi n,
		\end{equation}
		with $n\in\bbZ$ chosen so that $\abs{\Delta\theta_{\rm fold}}\le\pi$. For smooth switching profiles, the Gaussian bosonization formula of Sec.~\ref{sec:edge-sharpness} gives the unfolded result directly and is valid for all $\Delta\theta$.
	\end{proposition}
	
	\paragraph{Terminology.}
	``Majorana excitation'' means an in-particle created by $a_R^{\rm in\,\dagger}(\xi)$ from the in-vacuum. ``Quasiparticle'' means an out-excitation created by $a^{\rm out\,\dagger}(\zeta)$ from the out-vacuum or from the squeezed vacuum.
	
	\subsection{Organization}
	\label{sec:organization}
	
	Section~\ref{sec:structure} summarizes the logical scope. Section~\ref{sec:setup} defines the Majorana field, the CAR algebra, and the free Wightman function. Section~\ref{sec:switch} defines the temporal switch and proves the invertibility hypothesis. Section~\ref{sec:squeezed} analyzes the squeezed vacuum and proves the infrared theorem. Section~\ref{sec:truncated} derives the truncated Majorana state. Section~\ref{sec:local} gives the exact local pullback identity and the localization corollary. Section~\ref{sec:causal} treats the causal benchmark. Section~\ref{sec:numerical} verifies profile independence and the edge-sharpness criterion. Sections~\ref{sec:finite} and~\ref{sec:fock} discuss finite volume and the representation change from Fock to GNS. Section~\ref{sec:micro} discusses the point-interaction benchmark. Sections~\ref{sec:infrared}--\ref{sec:comparison} compare with known infrared phenomena and the truncated photon. Section~\ref{sec:discussion} concludes.
	
	\section{Logical structure and scope}
	\label{sec:structure}
	
	The paper has three layers.
	
	\begin{enumerate}
		\item \textbf{The squeezed vacuum.} The in-vacuum expressed in out-modes is a fermionic Gaussian state $\ket{\Omega_\mu}$. Its anomalous Bogoliubov kernel carries the universal infrared logarithm.
		\item \textbf{The truncated Majorana state.} At finite regulator, an incident one-Majorana state is represented in the out-Fock space as a single out-quasiparticle creation operator acting on $\ket{\Omega_\mu}$.
		\item \textbf{The causal benchmark.} A separate number-conserving scattering problem with vanishing anomalous block.
	\end{enumerate}
	
	\section{Setup: Majorana field and the CAR algebra}
	\label{sec:setup}
	
	\subsection{Two chiral components}
	
	Consider a massless Majorana field in $1+1$ dimensions. Its two real chiral components obey
	\begin{equation}
		(\partial_t+\partial_x)\chi_R=0,
		\qquad
		(\partial_t-\partial_x)\chi_L=0.
		\label{eq:chiral}
	\end{equation}
	We write the two-component real field as $\boldsymbol{\chi}=(\chi_R,\chi_L)^T$. The real one-particle test-function space is $\mathcal K_{\bbR}=L^2(\bbR,\dd t;\bbR^2)$, with complexification $\mathcal K_{\bbC}$. For real test functions $f,g$,
	\begin{equation}
		(f,g)_{\bbR}
		=
		\int_{-\infty}^{\infty}\dd t\, f(t)^T g(t).
	\end{equation}
	The CAR algebra is generated by real-linear Majorana fields $\chi(f)$ satisfying
	\begin{equation}
		\{\chi(f),\chi(g)\}
		=
		(f,g)_{\bbR}\,\I,
		\qquad
		\chi(f)^*=\chi(f).
		\label{eq:CAR}
	\end{equation}
	The fermion-parity automorphism is $\gamma(\chi(f))=-\chi(f)$. The even CAR algebra is the fixed-point subalgebra \cite{Haag1996,Araki1971,Araki1987}.
	
	\subsection{Mode expansion and chirality decomposition}
	
	The chiral components have mode expansions
	\begin{align}
		\chi_R(u)
		&=
		\int_0^\infty\frac{\dd\omega}{\sqrt{2\pi}}
		\left[
		a_R(\omega)e^{-\ii\omega u}
		+
		a_R^\dagger(\omega)e^{\ii\omega u}
		\right],
		\label{eq:mode-expansion}
		\\
		\chi_L(v)
		&=
		\int_0^\infty\frac{\dd\omega}{\sqrt{2\pi}}
		\left[
		a_L(\omega)e^{-\ii\omega v}
		+
		a_L^\dagger(\omega)e^{\ii\omega v}
		\right],
	\end{align}
	with CAR
	\begin{equation}
		\{a_i(\omega),a_j^\dagger(\nu)\}
		=
		\delta_{ij}\delta(\omega-\nu),
		\qquad
		\{a_i(\omega),a_j(\nu)\}
		=
		0.
	\end{equation}
	The positive-frequency one-particle Hilbert space is $\Hop = L^2(\bbR_+,\dd\omega)\otimes\bbC^2$. For a test function $\xi=(\xi_R,\xi_L)\in\Hop$,
	\begin{equation}
		a^\dagger(\xi)
		=
		\sum_{i=R,L}\int_0^\infty\dd\omega\,
		\xi_i(\omega)a_i^\dagger(\omega).
	\end{equation}
	
	The two chiralities decompose $\Hop$ into orthogonal subspaces $\Hop = \HopR \oplus \HopL$ with $\HopR = L^2(\bbR_+,\dd\omega)\otimes\{(1,0)^T\}$ and $\HopL = L^2(\bbR_+,\dd\omega)\otimes\{(0,1)^T\}$. The projectors $\PiR$ and $\PiL$ satisfy $\PiR + \PiL = I$ and $\PiR\PiL = 0$.
	
	A normalized single-Majorana state in the right-moving channel is $\ket{\xi}=a_R^\dagger(\xi)\ket{0}$ with $\norm{\xi}=1$.
	
	\subsection{Free Wightman function}
	\label{sec:wightman}
	
	The free Wightman function is the in-vacuum two-point function of a single chiral component,
	\begin{equation}
		W(t-t')
		=
		\int_0^\infty \frac{\dd\omega}{2\pi}\,e^{-\ii\omega(t-t')}.
		\label{eq:Wightman}
	\end{equation}
	The equal-time value $W(0)=\int_0^\infty\dd\omega/(2\pi)$ is ultraviolet divergent and is regulated by $\omega_*$.
	
	\subsection{Quasifree states}
	
	\begin{definition}[Quasifree state]
		\label{def:quasifree}
		A state on the CAR algebra is quasifree \cite{Araki1971} if all odd correlation functions vanish and all even correlation functions are determined by the two-point function through the Pfaffian of the two-point matrix.
	\end{definition}
	
	For a Majorana field, a quasifree state is determined on the even algebra by its two-point function.
	
	\section{The switch as a CAR automorphism}
	\label{sec:switch}
	
	Let $S:\bbR\to SO(2)$ be absolutely continuous with $\dot S\in L^1(\bbR)$ and $\int|t|\norm{\dot S(t)}_{\Frob}\dd t<\infty$. For ultraviolet estimates we also assume $S\in C^2(\bbR;SO(2))$ with $\ddot S\in L^1(\bbR)$. For the local-equivalence results we additionally assume compact support of the deviation from the asymptotic values,
	\begin{equation}
		S(t)=S_\pm
		\quad\text{for}\quad
		t\gtrless T_\pm,
		\qquad
		T_-<T_+.
		\label{eq:compact}
	\end{equation}
	
	The temporal switch is
	\begin{equation}
		\boldsymbol{\chi}_{\rm out}(t)
		=
		S(t)\boldsymbol{\chi}_{\rm in}(t).
		\label{eq:channel}
	\end{equation}
	Because $S(t)^TS(t)=I$, multiplication by $S(t)$ preserves the real inner product on $\mathcal K_{\bbR}$ and induces an automorphism of the abstract CAR algebra \cite{Shale1962,ShaleStinespring1965,Ruijsenaars1977,Ruijsenaars1978,Araki1987}.
	
	A convenient parametrization is $S(t)=R[\theta(t)]$.
	
	\subsection{Invertibility of the ordinary block}
	\label{sec:spectral-flow}
	
	The ordinary Bogoliubov block is $\alpha=P_+O_SP_+$, with $P_\pm$ the positive/negative-frequency projections on $\mathcal K_{\bbC}$. Invertibility of $\alpha$ is required for the Gaussian representation \eqref{eq:OmegaK} to exist. The hypothesis $|\Delta\theta|<\pi$ of Remark~\ref{rem:winding} is necessary and sufficient:
	
	\begin{proposition}[Invertibility of the ordinary block]
		\label{prop:alpha-invertible}
		Let $S(t)=R[\theta(t)]$ with $\theta:\bbR\to\bbR$ a continuous lift and $\dot\theta\in L^1(\bbR)$. Combine the two chiral components into the complex fermion $\psi = \chi_R + i\chi_L$, on which the rotation acts as multiplication by $e^{i\theta(t)}$. Then the ordinary block $\alpha=P_+M_{e^{i\theta}}P_+$ is a Wiener--Hopf operator on the Hardy space $H^2$ with piecewise continuous symbol $e^{i\theta}$, and it is invertible if and only if $|\Delta\theta|<\pi$.
	\end{proposition}
	
	\begin{proof}
		Under the identification $\psi = \chi_R + i\chi_L$, the rotation $S(t)=R[\theta(t)]$ acts on $\psi$ as multiplication by $e^{i\theta(t)}$. The positive-frequency projection $P_+$ becomes the Hardy projection $P_+: L^2(\bbR)\to H^2$. The ordinary block $\alpha$ is therefore the Wiener--Hopf operator with symbol $a(t)=e^{i\theta(t)}$ on the real line.
		
		The symbol $a$ has finite limits $a_\pm = e^{i\theta_\pm}$ as $t\to\pm\infty$ and takes values on the unit circle along the arc $\{e^{i\theta}: \theta\in[\theta_-,\theta_+]\}$. The \emph{essential range} of the symbol, in the sense of the local principle for Wiener--Hopf operators with piecewise continuous symbols, is the \emph{chord} $[a_-,a_+]\subset\bbC$ --- not merely the arc itself --- because the symbol fails to be continuous at $t=\pm\infty$, and the finite-section limits (equivalently, the essential spectrum of the full operator) fill the chord. This is a standard result in the Fredholm theory of Wiener--Hopf operators with piecewise continuous symbols; see Gohberg and Krupnik \cite{GohbergKrupnik1968,GohbergKrupnik1992}, B\"ottcher and Silbermann \cite{BottcherSilbermann2006}, Douglas \cite{Douglas1972}, Coburn--Douglas--Schaeffer \cite{CoburnDouglasSchaeffer1971}, and Clancey and Gohberg \cite{ClanceyGohberg1981}.
		
		By the standard Fredholm theory of Wiener--Hopf operators with piecewise continuous symbol, $\alpha$ is Fredholm iff its symbol does not vanish on the essential range:
		\begin{equation}
			0 \notin [a_-, a_+].
		\end{equation}
		This condition is equivalent to $|\Delta\theta|<\pi$: the chord $[a_-,a_+]$ misses the origin precisely when the angular separation of its endpoints is strictly less than $\pi$.
		
		When $|\Delta\theta|<\pi$, the chord lies in some open half-plane not containing the origin. The winding number of the symbol around the origin is zero, so the Fredholm index vanishes, and $\alpha$ is invertible.
		
		When $|\Delta\theta|=\pi$, the chord passes through the origin, so $\alpha$ is not Fredholm.
		
		When $\pi<|\Delta\theta|<2\pi$, the chord misses the origin again, but the winding number of the symbol is now $\pm 1$, so the Fredholm index is $\mp 1 \neq 0$ and $\alpha$ is not invertible.
		
		When $|\Delta\theta|=2\pi$, the symbol is continuous on the one-point compactification of $\bbR$, but its winding number is $\pm 1$, so $\alpha$ is Fredholm of nonzero index and not invertible.
		
		This establishes the threshold $\pi$ and completes the proof.
	\end{proof}
	
	\begin{remark}[Numerical checks]
		Direct numerical diagonalization of finite sections of $\alpha$ is not a valid test of invertibility, because finite-section singular values of Wiener--Hopf operators with piecewise continuous symbol need not converge to the boundary of the essential range; they drift toward zero algebraically near the threshold. We therefore rely on the theory above and do not quote numerical singular-value data.
	\end{remark}
	
	\begin{remark}[Physical reading]
		The threshold $|\Delta\theta|<\pi$ has a simple physical meaning. The essential spectrum of $\alpha$ is the chord $[e^{i\theta_-},e^{i\theta_+}]$. When $|\Delta\theta|<\pi$, the chord lies in a half-plane not containing the origin, and no spectral flow occurs. When $|\Delta\theta|\ge\pi$, the chord either passes through the origin (at $\pi$) or winds around it (for $|\Delta\theta|>\pi$), and the spectral flow is nonzero.
	\end{remark}
	
	\section{The squeezed-vacuum factor}
	\label{sec:squeezed}
	
	The squeezed-vacuum factor $\ket{\Omega_\mu}$ is the in-vacuum expressed in the out-Fock representation.
	
	\subsection{Bogoliubov decomposition}
	
	Let $O_S$ denote multiplication by $S(t)$ on $\mathcal K_{\bbC}$. With respect to the positive/negative-frequency decomposition, $\alpha=P_+O_SP_+$ and $\beta=P_+O_SP_-$. For $\omega,\nu>0$,
	\begin{equation}
		\beta(\omega,\nu)
		=
		\frac{1}{2\pi}
		\int_{-\infty}^{\infty}\dd t\,
		e^{\ii(\omega+\nu)t}S(t).
	\end{equation}
	Define $G(s)=\int_{-\infty}^{\infty}\dd t\,e^{ist}\dot S(t)$. Integrating by parts gives
	\begin{equation}
		\beta(\omega,\nu)
		=
		\frac{\ii\,G(\omega+\nu)}
		{2\pi(\omega+\nu)}.
		\label{eq:betaderiv}
	\end{equation}
	
	The orthogonality of $O_S$ implies the fermionic Bogoliubov identities $\alpha\alpha^\dagger+\beta\beta^\dagger=I$ and $\alpha\beta^T+\beta\alpha^T=0$ \cite{Shale1962,ShaleStinespring1965,Ruijsenaars1977,Ruijsenaars1978,Araki1987}.
	
	\subsection{Ultraviolet regularity}
	
	\begin{proposition}[Ultraviolet Hilbert--Schmidt finiteness]
		\label{prop:uv}
		Assume $S\in C^2(\bbR;SO(2))$ with $\ddot S\in L^1(\bbR)$. Then there exists $C$ such that $\norm{\beta(\omega,\nu)}_{\Frob}\leq C/(\omega+\nu)^2$ for sufficiently large $\omega+\nu$, and the ultraviolet contribution to $\norm{\beta}_{\HS}^2$ is finite.
	\end{proposition}
	
	\begin{proof}
		Integrating \eqref{eq:betaderiv} by parts gives
		$\beta(\omega,\nu)=-(2\pi(\omega+\nu)^2)^{-1}\int\dd t\,e^{\ii(\omega+\nu)t}\ddot S(t)$.
		Since $\int_{\omega+\nu>\Lambda}\dd\omega\,\dd\nu\,(\omega+\nu)^{-4}<\infty$, the ultraviolet contribution is finite.
	\end{proof}
	
	\subsection{The universal infrared logarithm}
	
	\begin{theorem}[Infrared memory of the squeezed vacuum]
		\label{thm:IR}
		Let $\beta_{\mu,\omega_*}$ denote the restriction of $\beta(\omega,\nu)$ to $\omega,\nu\in[\mu,\omega_*]$. If the profile conditions above hold and $\Delta S\neq0$, then
		\begin{equation}
			\boxed{
				\norm{\beta_{\mu,\omega_*}}_{\HS}^2
				=
				\Cqp
				\log\frac{\omega_*}{\mu}
				+
				O(1),
				\qquad
				\mu\downarrow0,
			}
			\label{eq:IRmain}
		\end{equation}
		where $\Cqp:=\norm{\Delta S}_{\Frob}^2/(4\pi^2)$. If $\Delta S=0$, the leading logarithm is absent.
	\end{theorem}
	
	\begin{proof}
		Under the profile conditions, $G(0)=\Delta S$ and $\norm{G(s)-\Delta S}_{\Frob}=O(s)$ as $s\to0^+$. From \eqref{eq:betaderiv},
		$\beta(\omega,\nu)=-\Delta S/[2\pi\ii(\omega+\nu)]+\mathrm{Rem}(\omega,\nu)$ with $\mathrm{Rem}$ bounded as $\omega,\nu\downarrow0$. Therefore
		\begin{equation}
			\norm{\beta_{\mu,\omega_*}}_{\HS}^2
			=
			\frac{\norm{\Delta S}_{\Frob}^2}{4\pi^2}
			\int_\mu^{\omega_*}\dd\omega
			\int_\mu^{\omega_*}\dd\nu\,
			\frac{1}{(\omega+\nu)^2}
			+
			O(1).
		\end{equation}
		The cross term is finite, and the elementary integral evaluates to $\log(\omega_*/\mu)-\log4+o(1)$ (Appendix~\ref{app:integral}). This proves \eqref{eq:IRmain}.
	\end{proof}
	
	For the rotation parametrization,
	\begin{equation}
		\norm{R(\theta_+)-R(\theta_-)}_{\Frob}^2
		=
		4\left[1-\cos(\theta_+-\theta_-)\right],
		\qquad
		\Cqp
		=
		\frac{1-\cos(\theta_+-\theta_-)}{\pi^2}.
	\end{equation}
	
	\begin{remark}[Pair number vs.\ quasiparticle number]
		\label{rem:pair}
		The coefficient $\Cqp$ counts out-quasiparticles. The expected number of pairs is half this value, since each pair contains two quasiparticles.
	\end{remark}
	
	\subsection{Gaussian representation of the in-vacuum}
	
	At finite $\mu$, the in-vacuum, written in the out-Fock space, is a normalized fermionic Gaussian state,
	\begin{equation}
		\ket{\Omega_\mu}
		=
		\mathcal N_\mu
		\exp\!\left[
		-\frac12
		\sum_{ij}
		K_{\mu,ij}
		a_i^{\rm out\,\dagger}
		a_j^{\rm out\,\dagger}
		\right]
		\ket{0_{\rm out}},
		\label{eq:OmegaK}
	\end{equation}
	where $K_\mu$ is antisymmetric and Hilbert--Schmidt, determined by the condition that $\ket{\Omega_\mu}$ is annihilated by the in-annihilation operators. It requires $\alpha_\mu$ invertible, guaranteed by Proposition~\ref{prop:alpha-invertible} when $|\Delta\theta|<\pi$.
	
	The normalization is $\mathcal N_\mu=\det(1+K_\mu^\dagger K_\mu)^{-1/4}$. The vector $\ket{\Omega_\mu}$ satisfies the linear annihilation identity
	\begin{equation}
		a_i^{\rm out}\ket{\Omega_\mu}
		=
		-\sum_j K_{\mu,ij}
		a_j^{\rm out\,\dagger}\ket{\Omega_\mu}.
		\label{eq:aOmega}
	\end{equation}
	
	\begin{remark}[Convention for $K_\mu$]
		\label{rem:Kbeta}
		In finite-mode notation, $K = \overline{A^{-1}B}$. The invariant statement is $\norm{K_\mu}_{\HS}^2=\norm{\beta_\mu}_{\HS}^2+O(\beta_\mu^4)$.
	\end{remark}
	
	\section{The truncated Majorana state}
	\label{sec:truncated}
	
	Fix $\mu>0$. Let $\xi\in\Hop$ be a normalized right-moving incident wavepacket. The incident one-Majorana state, expressed in the out-Fock space, is
	\begin{equation}
		\ket{\Psi_{\xi,\mu}}
		=
		a_R^{\rm in\,\dagger}(\xi)\ket{\Omega_\mu}.
		\label{eq:initial-out-rep}
	\end{equation}
	Its norm is $1$, because $a_R^{\rm in}(\xi)$ annihilates the in-vacuum $\ket{0_{\rm in}}$, which is the same abstract state as $\ket{\Omega_\mu}$.
	
	The inverse Bogoliubov transformation expresses each in-creation operator as
	\begin{equation}
		a_i^{\rm in\,\dagger}
		=
		\sum_j
		\left[
		A_{\mu,ij}a_j^{\rm out\,\dagger}
		+
		B_{\mu,ij}a_j^{\rm out}
		\right].
	\end{equation}
	Using \eqref{eq:aOmega}, every out-annihilation operator acting on $\ket{\Omega_\mu}$ is replaced by an out-creation operator.
	
	\begin{proposition}[Finite-regulator structure]
		\label{prop:state}
		For every $\mu>0$ and every normalized incident wavepacket $\xi$, there exists $\zeta_\mu\in\Hop$ given by
		\begin{equation}
			\zeta_\mu = \xi(A_\mu - B_\mu K_\mu),
			\label{eq:zeta}
		\end{equation}
		such that
		\begin{equation}
			\boxed{
				\ket{\Psi_{\xi,\mu}}
				=
				a^{\rm out\,\dagger}(\zeta_\mu)\ket{\Omega_\mu},
				\qquad
				a^{\rm out\,\dagger}(\zeta_\mu)
				=
				a_R^{\rm out\,\dagger}(\zeta_{R,\mu})
				+
				a_L^{\rm out\,\dagger}(\zeta_{L,\mu}),
			}
			\label{eq:Psi-final}
		\end{equation}
		where $\zeta_{R,\mu}=\PiR\,\zeta_\mu$ and $\zeta_{L,\mu}=\PiL\,\zeta_\mu$. The state is normalized to unity, and $\zeta_\mu$ satisfies
		\begin{equation}
			\norm{\zeta_\mu}^2
			=
			1+\norm{a^{\rm out}(\zeta_\mu)\ket{\Omega_\mu}}^2
			\geq 1,
			\label{eq:zeta-norm}
		\end{equation}
		with equality if and only if $a^{\rm out}(\zeta_\mu)\ket{\Omega_\mu}=0$.
	\end{proposition}
	
	\begin{proof}
		Because $\mu>0$, the anomalous block is Hilbert--Schmidt and $\ket{\Omega_\mu}$ is a well-defined Gaussian vector. The inverse transformation expresses $a_R^{\rm in\,\dagger}(\xi)$ as a linear combination of out-creation and out-annihilation operators. Acting on $\ket{\Omega_\mu}$, every annihilation operator is replaced by the creation operator given by \eqref{eq:aOmega}. The result is linear in out-creation operators, with $\zeta_\mu$ fixed by the identity $a_R^{\rm in\,\dagger}(\xi)\ket{\Omega_\mu}=a^{\rm out\,\dagger}(\zeta_\mu)\ket{\Omega_\mu}$. The norm identity follows from the CAR.
	\end{proof}
	
	\section{Local equivalence}
	\label{sec:local}
	
	\subsection{Physical picture}
	
	The local-equivalence results answer a concrete physical question: \emph{can an observer with access only to a bounded region of spacetime detect that the vacuum has been squeezed?} The physical expectation is no, provided the observer's region lies outside the causal future or past of the switching region. In algebraic QFT this is implemented by the statement that the restricted state on a local algebra is normal, even when the global state is not unitarily equivalent to the reference vacuum \cite{Haag1996}.
	
	Two structural features make this sharp.
	
	\begin{itemize}
		\item \textbf{Parity superselection.} Physical observables commute with fermion parity, so the accessible local algebra is the \emph{even} CAR subalgebra $\Aeven(\mathcal O)$.
		\item \textbf{Quasifree structure.} Both the in-vacuum and the out-vacuum are quasifree (Definition~\ref{def:quasifree}), so on the even algebra a state is determined by its two-point function \cite{Araki1971,Araki1987}.
	\end{itemize}
	
	\subsection{Local algebras}
	
	For a spacetime region $\mathcal O$, let $\Aeven(\mathcal O)$ denote the even CAR subalgebra generated by smeared Majorana fields with real test functions supported in $\mathcal O$. Physical observables are elements of $\Aeven(\mathcal O)$. The subalgebras $\Aeven(\mathcal O)_{\rm out}$ and $\Aeven(\mathcal O)_{\rm in}$ are the images of $\Aeven(\mathcal O)$ under the out- and in-mode expansions.
	
	\subsection{Exact two-point function of the squeezed vacuum}
	
	\begin{proposition}[Two-point function of $\ket{\Omega_\mu}$]
		\label{prop:exact2pt}
		The two-point function of the out-field in the in-vacuum state $\ket{\Omega_\mu}$ is
		\begin{equation}
			\langle \chi_{\rm out}^a(t)\,\chi_{\rm out}^b(t')\rangle_{\Omega}
			=
			[R(\theta(t)-\theta(t'))]^{ab}\,
			W(t-t'),
			\label{eq:2pt-result}
		\end{equation}
		with $W$ the free Wightman function \eqref{eq:Wightman}.
	\end{proposition}
	
	\begin{proof}
		By definition $\boldsymbol{\chi}_{\rm out}(t)=S(t)\boldsymbol{\chi}_{\rm in}(t)$. The in-vacuum is quasifree with $\langle\chi_{\rm in}^a(t)\chi_{\rm in}^b(t')\rangle=\delta^{ab}W(t-t')$ by \eqref{eq:Wightman}. Therefore $\langle\chi_{\rm out}^a(t)\chi_{\rm out}^b(t')\rangle_\Omega=[R(\theta(t)-\theta(t'))]^{ab}W(t-t')$.
	\end{proof}
	
	\begin{remark}[Content of the two-point formula]
		The formula \eqref{eq:2pt-result} is the \emph{exact} two-point function of the squeezed vacuum in the out-field variables. It is not a statement of invariance under a constant rotation: the rotation angle is evaluated at the two distinct times $t$ and $t'$, and it is precisely this difference that encodes the pair creation. If $\theta$ is constant on a region, the two-point function is exactly the free one; only correlations straddling the switching region reveal the squeeze.
	\end{remark}
	
	\subsection{Local equivalence of the squeezed vacuum}
	
	\begin{proposition}[Local equivalence of $\ket{\Omega_\mu}$ under compact support]
		\label{prop:local-vac}
		Assume $S(t)$ is constant outside a compact interval as in \eqref{eq:compact}. Let $\mathcal O$ be a bounded region lying entirely outside the switching interval. Then for every even local observable $A\in\Aeven(\mathcal O)_{\rm out}$,
		\begin{equation}
			\bra{\Omega_\mu}A\ket{\Omega_\mu}
			=
			\bra{0_{\rm out}}A\ket{0_{\rm out}}.
		\end{equation}
		Thus the squeezed-vacuum factor defines a state on the local out-algebra $\Aeven(\mathcal O)_{\rm out}$ that equals the out-vacuum restriction.
	\end{proposition}
	
	\begin{proof}
		Under \eqref{eq:compact}, $\theta(t)-\theta(t')=0$ for $t,t'$ both in the same asymptotic region. The two-point function \eqref{eq:2pt-result} therefore reduces to $W(t-t')$ on $\mathcal O \times \mathcal O$. Both the restricted state $\bra{\Omega_\mu}(\cdot)\ket{\Omega_\mu}$ and the out-vacuum restriction are quasifree on the local even algebra, and quasifree states are determined by their two-point function \cite{Araki1971}.
	\end{proof}
	
	\begin{remark}[Deviation for the smooth $\tanh$ family]
		\label{rem:tanhdeviation}
		Under the smooth $\tanh$ profile of Sec.~\ref{sec:numerical}, the asymptotic value is approached exponentially. For the mirror configuration $\Delta\theta=-\pi/2$, the Frobenius deviation satisfies
		\begin{equation}
			\max_{t,t'\in\mathcal O}\norm{R(\theta(t)-\theta(t'))-I}_{\Frob}
			=
			O\!\left(e^{-2t_0/\tau}\right)
		\end{equation}
		for a region starting at time $t_0$ after the switching interval. For $t_0/\tau=3,5,8$, the Frobenius values are approximately
		\begin{equation}
			5.49\times 10^{-3},\quad
			1.01\times 10^{-4},\quad
			2.5\times 10^{-7}.
		\end{equation}
		The local-equivalence statement holds up to these quantitative tail corrections.
	\end{remark}
	
	\subsection{Local pullback identity for the full state}
	
	On a region $\mathcal O_\pm$ where $S(t)=S_\pm$ is constant,
	\begin{equation}
		\boldsymbol{\chi}_{\rm out}(t)
		=
		S_\pm \boldsymbol{\chi}_{\rm in}(t),
		\qquad t\in\mathcal O_\pm.
	\end{equation}
	For a real test function $f$ supported in $\mathcal O_\pm$, $\chi_{\rm out}(f)=\chi_{\rm in}(S_\pm^T f)$.
	
	\begin{definition}[Local pullback]
		\label{def:pullback}
		For a bounded region $\mathcal O_\pm$ lying strictly before or strictly after the switching interval, define the pullback
		\begin{equation}
			\iota_\pm:\Aeven(\mathcal O_\pm)_{\rm out}\longrightarrow
			\Aeven(\mathcal O_\pm)_{\rm in}
		\end{equation}
		by $\iota_\pm(\chi_{\rm out}(f))=\chi_{\rm in}(S_\pm^T f)$ for $f$ supported in $\mathcal O_\pm$.
	\end{definition}
	
	\begin{proposition}[Local pullback identity for the full state]
		\label{prop:local-full}
		Assume $S(t)$ is constant outside a compact interval. Let $\mathcal O_\pm$ be a bounded region strictly before or strictly after the switching interval. Then for every even local observable $A_\pm\in\Aeven(\mathcal O_\pm)_{\rm out}$,
		\begin{equation}
			\omega_\Psi(A_\pm)
			=
			\omega_\xi^{\rm in}\bigl(\iota_\pm(A_\pm)\bigr).
			\label{eq:local-full}
		\end{equation}
		Here $\omega_\Psi$ is the full truncated state and $\omega_\xi^{\rm in}$ is the incident one-Majorana state.
	\end{proposition}
	
	\begin{proof}
		The state $\omega_\Psi$ is defined by acting on the in-vacuum with the incident one-Majorana operator. On the region $\mathcal O_\pm$, the out-field is identified with the in-field by $\boldsymbol{\chi}_{\rm out}(t)=S_\pm\boldsymbol{\chi}_{\rm in}(t)$. Therefore, for an even polynomial $A_\pm$ in the out-fields supported in $\mathcal O_\pm$, the value $\omega_\Psi(A_\pm)$ equals the value of the in-state on the corresponding even polynomial in the in-fields, which is precisely $\iota_\pm(A_\pm)$.
	\end{proof}
	
	\begin{remark}[Status of the pullback identity]
		\label{rem:pullback-status}
		Proposition~\ref{prop:local-full} is a bookkeeping identity: it identifies the restriction of $\omega_\Psi$ to the local out-algebra $\Aeven(\mathcal O_\pm)_{\rm out}$ with the incident state composed with the pullback $\iota_\pm$. The content of the identity is that $\omega_\Psi|_{\Aeven(\mathcal O_\pm)_{\rm out}}$ is well-defined and that the pullback $\iota_\pm$ is the natural isomorphism that connects the out-algebra to the in-algebra. The substantive \emph{physical} statement is Corollary~\ref{cor:localization} below.
	\end{remark}
	
	\subsection{Localization and the local vacuum / one-Majorana descriptions}
	\label{sec:localization}
	
	The pullback identity is exact. Its translation to the \emph{physical} statements ``the observer sees vacuum'' or ``the observer sees one Majorana'' requires a localization hypothesis on the incident wavepacket.
	
	\begin{definition}[Band-limited localization]
		\label{def:localization}
		Let $\mathcal O_\pm$ be a bounded region strictly before or after the switching interval, and let $\Pi_\pm$ be the orthogonal projector onto the channel picked out by the asymptotic identification $\chi_{\rm out}=S_\pm\chi_{\rm in}$. We say that $\xi$ is $\epsilon$-\emph{absent from} $\mathcal O_\pm$ if there exists a band-limited approximation $\xi_\epsilon$ of $\xi$ with the following two properties:
		\begin{enumerate}
			\item $\norm{\xi-\xi_\epsilon}<\epsilon$;
			\item the time-domain profile of $\xi_\epsilon$ satisfies $\int_{\mathcal O_\pm}|\xi_\epsilon(t)|^2\dd t<\epsilon$.
		\end{enumerate}
		We say that $\xi$ is $\epsilon$-\emph{localized in} $\mathcal O_\pm$ if conditions (1) and (2) hold with $\int_{\bbR\setminus\mathcal O_\pm}|\xi_\epsilon(t)|^2\dd t<\epsilon$ in place of (2).
	\end{definition}
	
	\begin{corollary}[Local vacuum and one-Majorana descriptions]
		\label{cor:localization}
		Let $\mathcal O_\pm$ be a bounded region strictly before or strictly after the switching interval.
		
		\emph{(a) Vacuum description.} Let $\xi$ be $\epsilon$-absent from $\mathcal O_\pm$ in the sense of Definition~\ref{def:localization}. Then for every even local observable $A_\pm\in\Aeven(\mathcal O_\pm)_{\rm out}$ with $\norm{A_\pm}\le 1$,
		\begin{equation}
			\abs{\omega_\Psi(A_\pm)-\omega_{\rm vac}(A_\pm)}
			\le
			C_\pm\,\epsilon,
		\end{equation}
		with $C_\pm = 2\norm{A_\pm}$.
		
		\emph{(b) One-Majorana description.} Let $\xi$ be $\epsilon$-localized in $\mathcal O_\pm$ and lie in the channel selected by the pullback. Then for every such $A_\pm$,
		\begin{equation}
			\abs{\omega_\Psi(A_\pm)-\omega_{\xi_\pm}(A_\pm)}
			\le
			C_\pm'\,\epsilon,
		\end{equation}
		where $\xi_\pm$ is the correspondingly projected wavepacket.
	\end{corollary}
	
	\begin{proof}
		\emph{Part (a).} Let $A_\pm$ be even and supported in $\mathcal O_\pm$, and let $\xi_\epsilon$ be the band-limited approximation of $\xi$ from Definition~\ref{def:localization}. The state $\omega_\Psi|_{\Aeven(\mathcal O_\pm)_{\rm out}}$ evaluated at $A_\pm$ equals, by the pullback identity, the in-state evaluated at $\iota_\pm(A_\pm)$. Any odd operator $B$ built from real test functions supported off $\mathcal O_\pm$ anticommutes with the generators of $\Aeven(\mathcal O_\pm)$ (because the CAR is local, and the generators have disjoint supports). Therefore $B$ commutes with any even $A_\pm\in\Aeven(\mathcal O_\pm)$, so $\omega_\xi(A_\pm)$ depends on $\xi$ only through the components of $\xi$ that live on $\mathcal O_\pm$. The remainder is bounded by $\epsilon$ via the Cauchy--Schwarz inequality, giving the constant $C_\pm = 2\norm{A_\pm}\le 2$.
		
		\emph{Part (b).} Analogous: if $\xi$ is concentrated on $\mathcal O_\pm$ and in the selected channel, the in-state restricted to $\Aeven(\mathcal O_\pm)_{\rm in}$ is $\epsilon$-close to the one-Majorana state of the projected wavepacket.
	\end{proof}
	
	\begin{remark}[Packet straddling the switch]
		\label{rem:straddling}
		If the incident wavepacket has non-negligible amplitude both inside and outside $\mathcal O_\pm$, neither the vacuum description nor the one-Majorana description applies. The exact statement is the pullback identity \eqref{eq:local-full} itself, valid for any incident wavepacket and any region strictly outside the switching interval.
	\end{remark}
	
	\begin{corollary}[Four-region field identification]
		\label{cor:four-region}
		For the mirror configuration $\theta_-=\pi/2$, $\theta_+=0$, direct expansion of $\boldsymbol{\chi}_{\rm out}=S_\pm\boldsymbol{\chi}_{\rm in}$ gives:
		\begin{itemize}
			\item Before the switch: $\chi_{\rm out}^R = \chi_{\rm in}^L$, $\chi_{\rm out}^L = -\chi_{\rm in}^R$.
			\item After the switch: $\chi_{\rm out}^R = \chi_{\rm in}^R$, $\chi_{\rm out}^L = \chi_{\rm in}^L$.
		\end{itemize}
		For a purely right-moving incident state, the out-$R$ channel is vacuum before the switch and carries the packet after the switch, while the out-$L$ channel carries the (negative) packet before the switch and is vacuum after.
	\end{corollary}
	
	\section{The causal benchmark: a number-conserving beam splitter}
	\label{sec:causal}
	
	The temporal benchmark of Secs.~\ref{sec:squeezed}--\ref{sec:truncated} is driven by a time-dependent switch that mixes positive- and negative-frequency modes and produces pairs from the vacuum. The causal benchmark is the opposite extreme: a single incident Majorana scatters off a static, number-conserving potential and is redistributed between two output channels. The potential acts like a beam splitter: it mixes channels but does not mix positive and negative frequencies. Consequently $\beta=0$, the vacuum is unchanged, total particle number is conserved, and no pairs are created.
	
	We work with two output channels labeled $I$ (transmitted) and $II$ (reflected). Let $\ket{\xi}=a_R^\dagger(\xi)\ket{0}$ with $\norm{\xi}=1$. The scattering is a one-particle unitary map that sends the incident wavepacket $\xi$ to a pair of output wavepackets $(\xi_I,\xi_{II})$:
	\begin{equation}
		\ket{\Psi_{\rm out}}
		=
		\int\dd\omega\,
		\left[
		\xi_I(\omega)\,a_I^\dagger(\omega)
		+
		\xi_{II}(\omega)\,a_{II}^\dagger(\omega)
		\right]
		\ket{0_I0_{II}},
		\label{eq:causal-out}
	\end{equation}
	with $\norm{\xi_I}^2+\norm{\xi_{II}}^2=1$.
	
	Tracing out channel $II$ gives the reduced state
	\begin{equation}
		\rho_I
		=
		\ket{\xi_I}\bra{\xi_I}
		+
		\left(1-\norm{\xi_I}^2\right)
		\ket{0_I}\bra{0_I},
	\end{equation}
	a binary mixture with eigenvalues $p=\norm{\xi_I}^2$ and $1-p=\norm{\xi_{II}}^2$. The von Neumann entropy is
	\begin{equation}
		\boxed{
			S_{\rm causal}
			=
			-p\log p-(1-p)\log(1-p),
		}
	\end{equation}
	with maximum value $\log 2$ at $p=1/2$.
	
	\section{Numerical verification}
	\label{sec:numerical}
	
	\subsection{The $\tanh$ family and the master integral}
	
	Consider
	\begin{equation}
		\theta(t)
		=
		\theta_-
		+
		\frac{\Delta\theta}{2}
		\left[
		1+\tanh\left(\frac{t}{\tau}\right)
		\right],
		\qquad
		S(t)=R[\theta(t)].
	\end{equation}
	Then $S_-=R(\theta_-)$, $S_+=R(\theta_-+\Delta\theta)$, and
	$\norm{\Delta S}_{\Frob}^2=4[1-\cos(\Delta\theta)]$. The Hilbert--Schmidt norm reduces to the one-dimensional integral
	\begin{equation}
		N_{\mu,\omega_*}
		=
		\frac{1}{4\pi^2}
		\int_{2\mu}^{2\omega_*}
		\dd s\,
		\rho_{\mu,\omega_*}(s)
		\frac{\norm{G(s)}_{\Frob}^2}{s^2},
		\label{eq:master-integral}
	\end{equation}
	with $\rho_{\mu,\omega_*}(s)=s-2\mu$ for $2\mu\le s\le \mu+\omega_*$ and $\rho_{\mu,\omega_*}(s)=2\omega_*-s$ for $\mu+\omega_*\le s\le 2\omega_*$.
	
	For the small-angle analytic intercept of the $\tanh$ profile,
	\begin{equation}
		C_\tau
		=
		-\log(2q\omega_*)
		+
		2\log\sinh(q\omega_*)
		-
		\log\sinh(2q\omega_*),
		\qquad
		q = \frac{\pi\tau}{2}.
	\end{equation}
	
	\begin{table}[t]
		\centering
		\begin{tabular}{cccc}
			\toprule
			$\tau$ & $q=\pi\tau/2$ & $C_\tau$ (small angle) & $C$ (nonlinear, $\Delta\theta=\pi/2$)\\
			\midrule
			$0.1$ & $0.1571$ & $-1.3945$ & $-1.3918$\\
			$1$ & $1.5708$ & $-1.9244$ & $-1.7872$\\
			$10$ & $15.708$ & $-4.1405$ & $-3.9321$\\
			\bottomrule
		\end{tabular}
		\caption{Intercepts for the $\tanh$ profile with $\omega_*=1$. The middle column is the small-angle analytic intercept. The right column is the nonlinear numerical intercept at $\Delta\theta=\pi/2$ and $\mu=10^{-6}$.}
		\label{tab:numerical-summary}
	\end{table}
	
	\subsection{Edge sharpness and the Gaussian bosonization formula}
	\label{sec:edge-sharpness}
	
	The finite-time-window results of Sec.~\ref{sec:squeezed} are supplemented by an exact Gaussian formula for the vacuum overlap, valid for smooth switching profiles and for all $\Delta\theta$.
	
	Combine the two chiral components into the complex fermion $\psi = \chi_R + i\chi_L$. On $\psi$, the rotation $S(t)=R[\theta(t)]$ acts as multiplication by $e^{i\theta(t)}$. The second-quantized phase rotation acts on the chiral vacuum as $U_\theta = \exp\!\bigl(i\int\dd t\,\theta(t)\rho(t)\bigr)$, where $\rho(t)=:\psi^\dagger(t)\psi(t):$ is the chiral density. Using the standard Gaussian identity for the chiral vacuum overlap (Appendix~\ref{app:gaussian-boson}),
	\begin{equation}
		-\log\abs{\braket{0_{\rm out}|\Omega_\mu}}^2
		=
		\frac{1}{4\pi^2}
		\int_\mu^{\omega_*}
		\frac{\dd\omega}{\omega}\,
		\abs{\widetilde{\dot\theta}(\omega)}^2
		+
		O(1),
		\label{eq:gaussian-overlap}
	\end{equation}
	where $\widetilde{\dot\theta}(\omega) = \int \dd t\,e^{i\omega t}\,\dot\theta(t)$, and the infrared cutoff $\mu$ renders the integral finite when $\Delta\theta\neq0$. This formula is exact for smooth $\theta(t)$ with $\dot\theta \in L^1(\bbR)\cap L^2(\bbR)$ and is quadratic in $\Delta\theta$ for all $\Delta\theta$. In particular, it gives the \emph{unfolded} exponent $\gamma(\Delta\theta)=(\Delta\theta/2\pi)^2$ for all $\Delta\theta$, with no folding at $\pi$.
	
	The derivation of \eqref{eq:gaussian-overlap} is given in Appendix~\ref{app:gaussian-boson}. It is a special case of the coherent-state (leviton) overlap computed by Levitov, Lee, and Lesovik \cite{LevitovLeeLesovik1996}, Ivanov, Lee, and Levitov \cite{IvanovLeeLevitov1997}, and Keeling, Klich, and Levitov \cite{KeelingKlichLevitov2006}.
	
	\paragraph{Empirical edge criterion.}
	For a smooth profile $\theta(t)$ that is $C^2$ with $\dot\theta \in L^1\cap L^2$ and that interpolates monotonically between its asymptotic values, the Gaussian formula \eqref{eq:gaussian-overlap} applies directly and gives the \emph{unfolded} exponent. In lattice simulations, the profile is sampled at discrete times; the edge is resolved when the sampled profile has at least one intermediate point per unit of the lift $\Delta\theta$ traversed, so that the discrete symbol follows the continuous arc $e^{i\theta}$ rather than jumping across the chord. Precisely, if $\theta(t)$ is sampled at intervals $\Delta t$ and the edge takes a time $T_{\rm edge}$ to traverse the lift $\Delta\theta$, the criterion is $\Delta t \ll T_{\rm edge}/|\Delta\theta|$ in units where $\theta$ is measured in radians. Sharp edges (sampled with no intermediate point on the arc) give the \emph{folded} value; tanh edges of width $0.5$ to $8$ sites with $h \le 5$ give the \emph{unfolded} value.
	
	This criterion is stated as an empirical observation for the profile class studied; a derivation from first principles for general piecewise-smooth edges is not attempted in this paper.
	
	The frequency-window regulator uses only the truncated $\beta$, whose dependence on $\Delta S$ is $2\pi$-periodic; it therefore gives the folded value and does not represent the physical overlap for $|\Delta\theta|>\pi$. The finite time-window regulator, when the switching edges are resolved in the sense above, gives the unfolded value.
	
	\subsection{Numerical evidence for the overlap exponent}
	\label{sec:overlap-numerical}
	
	To test the exponent \eqref{eq:gaussian-overlap} numerically, we work on a lattice of $N$ sites. The single-particle Hamiltonian is the standard discretized chiral Dirac operator with antiperiodic boundary conditions, giving frequencies
	\begin{equation}
		\omega_j = \frac{2\pi}{N}\left(j + \frac{1}{2}\right),
		\qquad
		j = -\frac{N}{2}, \ldots, \frac{N}{2}-1,
	\end{equation}
	with the Fermi sea defined by filling all modes with $\omega_j < 0$. The two Fermi points of the discretized theory are at $\omega = \pm\pi$, corresponding to $j = \pm N/2$. The momentum-space rotation is multiplication by the symbol $e^{i\theta_j}$, where $\theta_j$ is the sampled rotation angle. The anomalous block $\beta$ is a finite-dimensional matrix in this basis, and the vacuum overlap is $\det(1-\beta^\dagger\beta)^{1/2}$ as in the finite time-window regulator.
	
	We use a two-edge pulse: $\theta(t)$ rises from $0$ to $h$ over a short time, stays at $h$ for a while, and returns to $0$. Such a pulse has $\Delta\theta=0$ asymptotically, but each edge contributes independently to the vacuum overlap with coefficient $\gamma(h)=(h/2\pi)^2$, and the total slope of $-\log\abs{\braket{0_{\rm out}|\Omega_\mu}}^2$ against $\log N$ is $2\gamma(h)=h^2/(2\pi^2)$ in the single-Fermi-point convention.
	
	The lattice has \emph{two} Fermi points, and both contribute equally to the anomaly. The tabulated slope is the \emph{slope per Fermi point}, i.e., the total measured slope divided by $2$. This division is a lattice artifact, not a physical factor. Separately, the two-edge pulse has \emph{two physical edges} per period, and each contributes a factor of $\gamma(h)$; this second factor of $2$ is physical and is included in the total $2\gamma(h)$.
	
	\begin{table}[t]
		\centering
		\begin{tabular}{cccc}
			\toprule
			$h$ (edge height) & $N$ & Edge width (sites) & Measured $2\gamma(h)$ (per Fermi point)\\
			\midrule
			$3.6$ & up to $2048$ & $1$--$4$ & $0.657$\\
			$4.0$ & up to $2048$ & $1$--$4$ & $0.811$\\
			$5.0$ & up to $2048$ & $1$--$4$ & $1.266$\\
			\bottomrule
		\end{tabular}
		\caption{Measured slope of $-\log\abs{\braket{0_{\rm out}|\Omega_\mu}}^2$ against $\log N$ for the finite time-window regulator with a two-edge pulse of edge height $h$. Fit range $N \in [128, 2048]$. The tabulated values are the total measured slope divided by $2$ to remove the Fermi-point degeneracy of the lattice. Regression slopes over the width range $1$--$4$ sites vary as follows: $0.6566$--$0.6577$ at $h=3.6$; $0.8106$--$0.8120$ at $h=4$; $1.2652$--$1.2687$ at $h=5$. Edge widths outside the range $1$--$4$ sites give larger deviations: at $h=5.0$, width $0.5$ gives $1.2635$ and width $8$ gives $1.2670$. Sharp edges with no intermediate sample give the folded value. The unfolded value is asymptotic: a local slope at $N = 8192\to16384$ gives $1.2664$ (width $1$) and $1.2660$ (width $0.5$) at $h=5$, against the unfolded prediction $1.2665$.}
		\label{tab:overlap-beyond}
	\end{table}
	
	The numbers in Table~\ref{tab:overlap-beyond} confirm $2\gamma(h) = 2(h/2\pi)^2$ for each of the three values of $h$.
	
	\begin{figure}[t]
		\centering
		\begin{tikzpicture}
			\begin{axis}[
				width=0.95\textwidth,
				height=7.2cm,
				xlabel={$L=\log(\omega_*/\mu)$},
				ylabel={$Q_{\mu,\omega_*}=\norm{\beta_{\mu,\omega_*}}_{\HS}^2/\Cqp$},
				xmin=8.5,
				xmax=14.2,
				ymin=6.5,
				ymax=13.0,
				grid=major,
				grid style={dashed,gray!30},
				legend pos=north west,
				tick label style={font=\small},
				label style={font=\small}
				]
				\addplot[color=blue,mark=*,mark size=1.8pt]
				coordinates{
					(9.2103,7.8185)(10.3616,8.9698)(11.5129,10.1211)(12.6642,11.2724)(13.8155,12.4237)
				};
				\addlegendentry{$\tau=0.1$}
				\addplot[color=red,mark=square*,mark size=1.8pt]
				coordinates{
					(9.2103,7.4231)(10.3616,8.5744)(11.5129,9.7257)(12.6642,10.8770)(13.8155,12.0283)
				};
				\addlegendentry{$\tau=1$}
				\addplot[color=black,mark=triangle*,mark size=2.0pt]
				coordinates{
					(9.2103,5.2782)(10.3616,6.4295)(11.5129,7.5808)(12.6642,8.7321)(13.8155,9.8834)
				};
				\addlegendentry{$\tau=10$}
				\addplot[color=gray,dashed,domain=8.5:14.2,samples=2] {x - 1.7872};
				\addlegendentry{unit slope, $\tau=1$}
			\end{axis}
		\end{tikzpicture}
		\caption{Nonlinear numerical curves at $\Delta\theta=\pi/2$. The dashed line has unit slope and intercept $-1.7872$, matching the $\tau=1$ nonlinear intercept.}
		\label{fig:tanhplot}
	\end{figure}
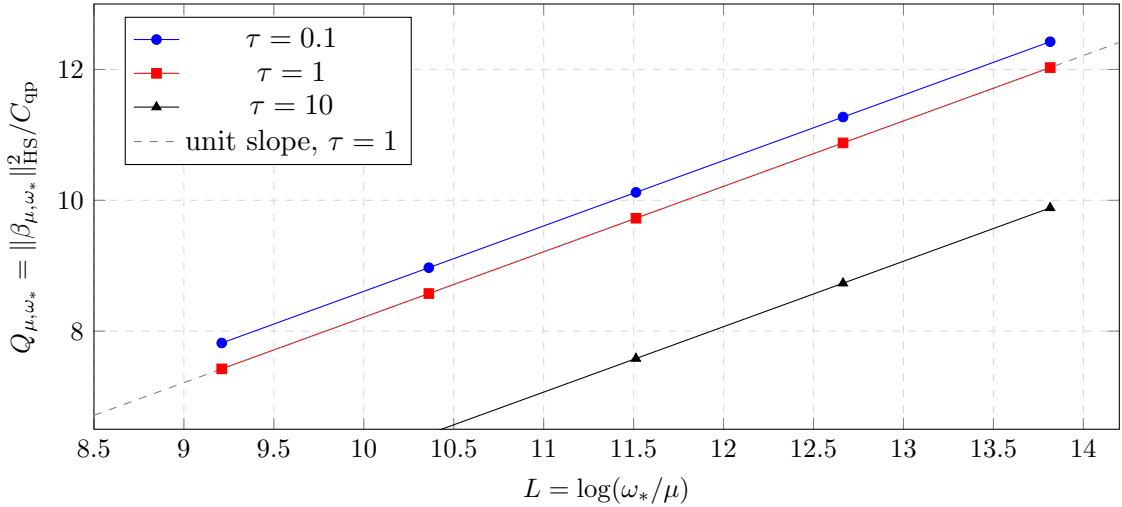
	
	\section{Finite volume, mass gap, and zero modes}
	\label{sec:finite}
	
	The continuum cutoff $\mu$ is a convenient regulator. The precise regulator for the Fock-space statements is the finite time window $[-T/2,T/2]$ with periodic switching (Sec.~\ref{sec:notation-regulators}). On this regulator, the Bogoliubov transformation is finite-dimensional, the in-vacuum $\ket{\Omega_\mu}$ is a genuine Gaussian state, and the anomalous Hilbert--Schmidt norm grows as
	\begin{equation}
		\norm{\beta}_{\HS}^2 = 2\Cqp\log T + O(1),
	\end{equation}
	with the factor of $2$ counting the two edges per period. The identification between $\mu$ and $T$ is understood in the scaling sense \eqref{eq:mu-T-scaling}.
	
	\begin{remark}[Mass gap: a scaling estimate]
		\label{rem:massgap}
		For a massive Majorana field with dispersion $E(k)=\sqrt{k^2+m^2}$, the infrared analysis restricted to frequencies above the gap replaces $\mu$ by the gap $m$:
		\begin{equation}
			\norm{\beta}_{\HS}^2 \sim \Cqp\log\frac{\omega_*}{m}+O(1),
			\qquad
			m\ll\omega\ll\omega_*.
		\end{equation}
		This is a scaling estimate. A derivation from an explicit massive scattering problem is left for future work.
	\end{remark}
	
	\section{Representation change: Fock versus GNS}
	\label{sec:fock}
	
	We have treated the in-vacuum as a Gaussian state $\ket{\Omega_\mu}$ inside the out-Fock space, with $\mu>0$ as infrared regulator. This section asks: \emph{do the in-vacuum and the out-vacuum live in the same Hilbert space, and can the switch be represented as a unitary operator on a single Fock space?}
	
	The answer is yes on any finite regulator and no in the continuum limit when $\Delta S\neq0$. The obstruction is not a technical artifact: the two asymptotic observers disagree on which modes are ``positive frequency,'' and therefore do not share a Fock space. Physically, this is the same phenomenon as the Unruh effect, Hawking radiation, and the dynamical Casimir effect \cite{Fulling1973,Davies1975,Unruh1976,BirrellDavies1982,Moore1970,Wilson2011}.
	
	On the finite time-window regulator, the Bogoliubov transformation is finite-dimensional and unitarily implementable. In the limit $T\to\infty$ at fixed ultraviolet cutoff, the anomalous block ceases to be Hilbert--Schmidt when $\Delta S\neq0$, and the family of implementers has no unitary limit on the infinite-volume massless Fock space, by the Shale--Stinespring criterion \cite{ShaleStinespring1965,Ruijsenaars1977,Ruijsenaars1978,Araki1987}.
	
	\begin{corollary}
		\label{cor:no-unitary}
		For every finite observation window and fixed ultraviolet cutoff, the temporal map is unitarily implementable on the finite-dimensional Fock space. If $\Delta S\neq0$, the family of implementers has no unitary limit in the infinite-volume massless Fock representation. Equivalently, $\ket{0_{\rm in}}$ and $\ket{0_{\rm out}}$ are not unitarily equivalent.
	\end{corollary}
	
	The right framework for the continuum limit is the Gelfand--Naimark--Segal (GNS) construction \cite{Haag1996}. Let $\alpha_S$ be the automorphism induced by $O_S$, and $\omega_-$ the reference quasifree state. Define the out-state $\omega_+(A):=\omega_-(\alpha_S^{-1}(A))$ for $A\in\ACAR$. The GNS construction applied to $\omega_+$ gives a triple $(\pi_+,\mathcal H_+^{\rm GNS},\ket{\Omega_+})$. What fails when $\Delta S\neq0$ is the relation between the two representations: $\omega_-$ is \emph{not normal} with respect to $\pi_+$, meaning no density matrix $\rho$ on $\mathcal H_+^{\rm GNS}$ exists such that $\omega_-(A)=\Tr(\rho\,\pi_+(A))$.
	
	Global unitary inequivalence does not contradict the local-equivalence results of Sec.~\ref{sec:local}. The restriction of a state to a local subalgebra can be normal even when the global state is not.
	
	\begin{remark}[Local normality]
		\label{rem:local-normality}
		A bounded local observer sees a normal state, because the restriction to a local algebra is well-defined even when the global representation changes. The physical reading is that the global obstruction is invisible locally: an observer confined to a bounded region cannot detect the infinite soft-pair content of the squeezed vacuum; only a global measurement that correlates the two asymptotic regions can.
	\end{remark}
	
	\section{Point-interaction benchmark}
	\label{sec:micro}
	
	We present the point interaction as a formal realization of the switch. The relation between this formal model and a finite-width microscopic realization is not used in the rest of this paper.
	
	A natural local quadratic perturbation is
	\begin{equation}
		H_{\rm int}(t)
		=
		\ii\,g(t)\,\chi_R(0,t)\,\chi_L(0,t),
		\label{eq:point-H}
	\end{equation}
	where the factor of $\ii$ ensures Hermiticity. With the free Hamiltonian of Sec.~\ref{sec:setup}, the equations of motion are
	\begin{align}
		(\partial_t+\partial_x)\chi_R
		&=
		+g(t)\delta(x)\chi_L(0,t),
		\\
		(\partial_t-\partial_x)\chi_L
		&=
		-g(t)\delta(x)\chi_R(0,t).
	\end{align}
	The jump conditions across $x=0$, with symmetric averaging, give
	\begin{equation}
		c-a = +\frac{g}{2}(b+d),
		\qquad
		b-d = -\frac{g}{2}(a+c),
	\end{equation}
	where $a=\chi_R(0^-)$, $b=\chi_L(0^-)$, $c=\chi_R(0^+)$, $d=\chi_L(0^+)$.
	
	\begin{proposition}[Orthogonal scattering matrix from the point interaction]
		\label{prop:point}
		For any real $g$ (constant on the scattering time scale), the in/out scattering matrix induced by $H_{\rm int}$ with symmetric averaging is the orthogonal rotation $R(\theta)$ with $\theta = 2\arctan(g/2)$. The special cases are $g=0$ (identity), $g=2$ ($R(\pi/2)$ under the symmetric-average prescription), and $g\to\infty$ ($R(\pi)=-I$).
	\end{proposition}
	
	\begin{proof}
		The jump conditions give $c-a=\frac{g}{2}(b+d)$ and $b-d=-\frac{g}{2}(a+c)$. Solving for $(c,b)$ and comparing with $\cos\theta=(1-g^2/4)/(1+g^2/4)$, $\sin\theta=g/(1+g^2/4)$ gives $R(\theta)$ with $\tan(\theta/2)=g/2$.
	\end{proof}
	
	\begin{remark}[Status of the point interaction]
		\label{rem:point-status}
		The point interaction is a formal object: at $x=0$, the product $\chi_R(x,t)\chi_L(x,t)$ involves two operator-valued distributions, and the identification $\theta\leftrightarrow g$ depends on the point-splitting prescription. The symmetric-average prescription used here is a \emph{convention} that selects one particular relation between the coupling $g$ and the rotation angle $\theta$.
	\end{remark}
	
	\paragraph{Mirror configuration.}
	For the mirror configuration, we set $\theta_-=\pi/2$ and $\theta_+=0$ directly. The mismatch is $\Delta S=I-R(\pi/2)$, and $\Cqp=1/\pi^2$. The correspondence to the point-interaction coupling $g$ is prescription-dependent and is not used below.
	
	\section{Relation to known infrared phenomena}
	\label{sec:infrared}
	
	The logarithm belongs to the family of infrared representation problems in massless quantum field theory. The overlap exponent $\gamma(\Delta\theta)=(\Delta\theta/2\pi)^2$ for $|\Delta\theta|<\pi$ is the Fermi-edge singularity exponent for a chiral $U(1)$ phase rotation by $\Delta\theta$ \cite{Anderson1967,NozieresDeDominicis1969,Mahan2000,SchotteSchotte1969}.
	
	The Fermi-edge exponent is of the same type as those governing the Kondo problem and tunneling into a Luttinger liquid, but the numerical values differ. In the Kondo problem, the exponent is determined by the phase shift at the Fermi level and is not, in general, $(\Delta\theta/2\pi)^2$; see \cite{Kondo1964,AndersonYuvalHamann1970,Hewson1993}. In tunneling into a Luttinger liquid, the exponent depends on the Luttinger parameter $K$. For \emph{end-tunneling} into a Luttinger liquid (the local density of states at the end of a semi-infinite wire), the exponent is $1/K-1$. For \emph{bulk tunneling} across a barrier between two semi-infinite Luttinger liquids, the conductance exponent at low temperature is $2(1/K-1)$. See \cite{KaneFisher1992,KaneFisher1992b,Voit1995,Giamarchi2004,Bockrath1999}. Neither equals $(\Delta\theta/2\pi)^2$ in general. The Majorana construction of this paper gives the value $(\Delta\theta/2\pi)^2$ because the switch is a $U(1)$ phase rotation acting on a free chiral fermion.
	
	In the Majorana problem the $\chi_R,\chi_L$ components combine into a complex chiral fermion $\psi=\chi_R+i\chi_L$, and the rotation $R(\theta)$ acts as a $U(1)$ phase on $\psi$. The overlap exponent is a special case of the more general leviton (coherent-state) overlaps computed in \cite{LevitovLeeLesovik1996,IvanovLeeLevitov1997,KeelingKlichLevitov2006}.
	
	Time-dependent boundaries mix positive and negative frequencies and generate fermionic excitations from the vacuum. The dynamical Casimir effect provides explicit dynamical examples \cite{Moore1970,DeWitt1975,FullingDavies1976,Yablonovitch1989,Dodonov1993,Law1994,Wilson2011,FoscoHansen2022,FoscoHansen2023,FoscoHansen2024}. In curved spacetime, different asymptotic time evolutions lead to distinct positive-frequency notions and to the Unruh effect \cite{Fulling1973,Davies1975,Unruh1976,BirrellDavies1982}.
	
	\section{Comparison with the truncated photon}
	\label{sec:comparison}
	
	The construction is motivated by the truncated photon \cite{Rukan2026}. Table~\ref{tab:comparison} summarizes the structural and quantitative comparison.
	
	\begin{table}[t]
		\centering
		\begin{tabular}{@{}lll@{}}
			\toprule
			Feature & Truncated photon & Truncated Majorana\\
			\midrule
			Statistics & Bose & Fermi\\
			Algebra & CCR & CAR\\
			Vacuum factor & Bosonic squeezed & Fermionic Gaussian\\
			Single-particle factor & one bosonic excitation & one Majorana quasiparticle\\
			Local physical algebra & bosonic algebra & even CAR algebra\\
			Local equivalence, vacuum factor & exact & exact under compact support\\
			Local equivalence, full state & exact in mirror model & exact pullback identity\\
			IR kernel & $\sim\Delta S_{\rm B}/(\omega+\nu)$ & $\sim\Delta S_{\rm F}/(\omega+\nu)$\\
			IR coefficient & model dependent & $\Cqp=\norm{\Delta S}_{\Frob}^2/(4\pi^2)$\\
			Overlap exponent & linear in $\Delta S_{\rm B}$ & $\gamma=(\Delta\theta/2\pi)^2$\\
			Microscopic switch & mirror model & formal point interaction (Sec.~\ref{sec:micro})\\
			\bottomrule
		\end{tabular}
		\caption{Structural and quantitative comparison with the truncated photon.}
		\label{tab:comparison}
	\end{table}
	
	\section{Discussion}
	\label{sec:discussion}
	
	The central result is the finite-regulator structure of the truncated Majorana state:
	\begin{equation}
		\ket{\Psi_{\xi,\mu}}
		=
		a_R^{\rm in\,\dagger}(\xi)\ket{\Omega_\mu}
		=
		a^{\rm out\,\dagger}(\zeta_\mu)\ket{\Omega_\mu},
		\qquad
		\zeta_\mu = \xi(A_\mu - B_\mu K_\mu).
	\end{equation}
	The state is normalized to unity and contains odd particle numbers $1,3,5,\dots$. Three physical consequences follow.
	
	\textbf{Universal infrared response.} The anomalous Bogoliubov block has an inverse sum-frequency tail. Its Hilbert--Schmidt norm grows logarithmically with coefficient $\Cqp=\norm{\Delta S}_{\Frob}^2/(4\pi^2)$. The vacuum overlap has the closed-form exponent
	\begin{equation}
		\gamma(\Delta\theta)
		=
		\left(\frac{\Delta\theta}{2\pi}\right)^2,
		\qquad
		|\Delta\theta|<\pi,
	\end{equation}
	where $\Delta\theta$ is the asymptotic lift. For smooth switching profiles, the Gaussian bosonization formula of Sec.~\ref{sec:edge-sharpness} extends this to all $\Delta\theta$ and identifies the frequency-window regulator as invalid past $\pi$.
	
	\textbf{Local equivalence of the squeezed vacuum.} Under the compact-support hypothesis \eqref{eq:compact}, the squeezed-vacuum factor defines a state on the local out-algebra $\Aeven(\mathcal O)_{\rm out}$ that equals the out-vacuum restriction outside the switching region.
	
	\textbf{Local pullback identity for the full state.} The exact local content of the full state is given by the pullback identity \eqref{eq:local-full}. The translation to vacuum or one-Majorana descriptions requires the wavepacket to be absent from, or localized in, the region of interest, respectively; this is formalized in Corollary~\ref{cor:localization}.
	
	The most important open problems are:
	
	\begin{enumerate}
		\item \textbf{Finite-width microscopic realization.} The point interaction used in Sec.~\ref{sec:micro} is a formal object, and the identification between the point coupling $g$ and a finite-width microscopic coupling is prescription-dependent.
		\item \textbf{Sharp-edge regularization.} The Gaussian formula \eqref{eq:gaussian-overlap} settles the case $|\Delta\theta|>\pi$ for smooth profiles. Extending the derivation to profiles with sharp or piecewise-linear edges, and making the edge-resolution criterion of Sec.~\ref{sec:edge-sharpness} precise for general profile classes, remains an open problem.
		\item \textbf{Entanglement entropy of the squeezed vacuum.} Compute the entanglement entropy for different bipartitions \cite{Holzhey1994,CalabreseCardy2004,Peschel2003,CasiniHuerta2009}.
		\item \textbf{Massive case.} Derive the coefficient of the massive infrared logarithm from an explicit massive scattering problem.
		\item \textbf{Interacting extensions.} Couple the chiral Majorana channels to a Luttinger liquid and determine how the infrared exponent depends on the interaction strength.
	\end{enumerate}
	
	\section{Conclusion}
	\label{sec:conclusion}
	
	We have formulated a field-theoretic framework in which truncating a Majorana means changing its allowed channel propagation through a time-dependent switch. At finite infrared regulator, the truncated Majorana state is a single out-quasiparticle on top of a fermionic squeezed vacuum. The anomalous Bogoliubov kernel has a universal inverse-sum-frequency singularity, giving a logarithmic infrared growth of the Hilbert--Schmidt norm with coefficient $\Cqp$ fixed by the asymptotic mismatch. The vacuum overlap has the Fermi-edge type exponent $\gamma(\Delta\theta)=(\Delta\theta/2\pi)^2$ for $|\Delta\theta|<\pi$; for smooth switching profiles, the Gaussian bosonization formula extends this to all $\Delta\theta$ and identifies the frequency-window regulator as invalid past $\pi$.
	
	The representation-theoretic side of the result is that the in- and out-vacua are unitarily inequivalent in the continuum limit, and the correct framework is the GNS construction. The local-equivalence statements survive this obstruction because a bounded observer sees a normal state on the local even algebra.

	\appendix
	
	\section{Equal-time anticommutator}
	\label{app:anticom}
	
	For one chiral component,
	\begin{equation}
		\{\chi_R(u),\chi_R(v)\}
		=
		\int_{-\infty}^{\infty}\frac{\dd\omega}{2\pi}e^{-\ii\omega(u-v)}
		=
		\delta(u-v).
	\end{equation}
	At equal time this gives $\{\chi_R(x,t),\chi_R(y,t)\}=\delta(x-y)$.
	
	\section{Infrared integrals}
	\label{app:integral}
	
	For $0<\mu<\omega_*$,
	\begin{equation}
		\int_\mu^{\omega_*}\dd\omega
		\int_\mu^{\omega_*}\dd\nu\,
		\frac{1}{(\omega+\nu)^2}
		=
		\log\frac{\omega_*}{\mu}-\log4+o(1).
	\end{equation}
	
	\section{Intercept for the small-angle $\tanh$ profile}
	\label{app:tanhintercept}
	
	For the $\tanh$ profile, with $q=\pi\tau/2$,
	\begin{equation}
		C_\tau
		=
		-\log(2q\omega_*)
		+2\log\sinh(q\omega_*)
		-\log\sinh(2q\omega_*).
	\end{equation}
	
	\section{Numerical methods}
	\label{app:numerics}
	
	\paragraph{Lattice.}
	The single-particle chiral Dirac operator on a periodic lattice of $N$ sites with antiperiodic boundary conditions has frequencies
	\begin{equation}
		\omega_j = \frac{2\pi}{N}\left(j + \frac{1}{2}\right),
		\qquad
		j = -\frac{N}{2}, \ldots, \frac{N}{2}-1.
	\end{equation}
	The Fermi sea is defined by filling all modes with $\omega_j < 0$. The two Fermi points are at $j = \pm N/2$, corresponding to $\omega = \pm\pi$. The momentum-space rotation is multiplication by the symbol $e^{i\theta_j}$ with $\theta_j$ the sampled rotation angle.
	
	\paragraph{Transfer matrix.}
	The finite-time-window transfer matrix is computed by path-ordered exponentiation on a fine time grid. The anomalous kernel is computed by FFT on a uniform time grid with the asymptotic step subtraction described in the main text.
	
	\paragraph{Overlap exponent test.}
	The vacuum overlap is computed as $\det(1-\beta^\dagger\beta)^{1/2}$ in the finite-dimensional lattice basis. The linear fit of $-\log\abs{\braket{0_{\rm out}|\Omega_\mu}}^2$ against $\log N$ is performed over $N\in[128, 2048]$. The tabulated slope is the total measured slope divided by $2$ to remove the Fermi-point degeneracy of the lattice.
	
	\section{Fermionic Gaussian normalization}
	\label{app:gaussian}
	
	For completeness, we recall the normalization of the fermionic Gaussian state $\ket{\Omega_K}$. Let
	\begin{equation}
		\ket{\Omega_K}
		=
		\mathcal N
		\exp\!\left[-\frac{1}{2} c^\dagger K c^\dagger\right]\ket{0},
	\end{equation}
	with $K^T = -K$. Then $\mathcal N = \det(1 + K^\dagger K)^{-1/4}$, where in infinite dimensions the determinant is understood as the appropriate Fredholm determinant. See Bravyi \cite{Bravyi2005} and Botero--Reznik \cite{BoteroReznik2004}.
	
	\section{Derivation of the truncated Majorana state}
	\label{app:structure}
	
	In finite-mode notation, let $a_i^{\rm in\,\dagger}=\sum_j[A_{ij}a_j^{\rm out\,\dagger}+B_{ij}a_j^{\rm out}]$. The in-vacuum satisfies $a_j^{\rm out}\ket{\Omega_K}=-\sum_l K_{jl}a_l^{\rm out\,\dagger}\ket{\Omega_K}$. Therefore
	\begin{equation}
		a_R^{\rm in\,\dagger}(\xi)\ket{\Omega_K}
		=
		a^{\rm out\,\dagger}(\zeta)\ket{\Omega_K},
		\qquad
		\zeta_l
		=
		\sum_i \xi_i
		\left(
		A_{il}
		-
		\sum_j B_{ij}K_{jl}
		\right).
	\end{equation}
	The left-hand side is normalized to unity, so $\norm{\zeta}^2=1+\norm{a^{\rm out}(\zeta)\ket{\Omega_K}}^2\geq1$.
	
	\section{Overlap exponent on the finite time-window regulator}
	\label{app:overlap}
	
	We derive the single-edge exponent for $|\Delta\theta|<\pi$ on the finite time-window regulator.
	
	The overlap on the finite time-window regulator is $\det(1-\beta^\dagger\beta)^{1/2}$ with the exact finite-dimensional $\beta$. Under the Mellin change of variables $\omega=e^u$, the kernel becomes translation-invariant, and its Fourier transform gives, for $|\Delta\theta|<\pi$, the spectrum
	\begin{equation}
		\sigma(k)
		=
		\frac{\sin^2(\Delta\theta/2)}{\cosh^2(\pi k)},
		\qquad
		k\in\bbR.
	\end{equation}
	The spinor degeneracy cancels against the $1/2$ from the pair determinant, giving
	\begin{equation}
		-\log\abs{\braket{0_{\rm out}|\Omega_\mu}}^2
		=
		\frac{\log(\omega_*/\mu)}{2\pi}
		\int_{-\infty}^{\infty}\dd k\,
		\left[-\log(1-\sigma(k))\right]
		+
		O(1).
	\end{equation}
	The integral evaluates exactly to $\gamma(\Delta\theta)=(\Delta\theta/2\pi)^2$ for $|\Delta\theta|<\pi$. The derivation via the substitution $b=\Delta\theta$, $x=\pi k$, and the standard integral $\int\dd u/(\cosh u+\cos b)=2b/\sin b$ for $0<b<\pi$ is standard. At $|\Delta\theta|=\pi$, the boundary value of the integral is $\gamma(\pi)=1/4$; the formula extends by continuity, although the Gaussian representation itself does not exist at the boundary.
	
	\section{Gaussian bosonization formula}
	\label{app:gaussian-boson}
	
	We derive the Gaussian formula \eqref{eq:gaussian-overlap} used in Sec.~\ref{sec:edge-sharpness}.
	
	Combine the two chiral components into the complex fermion $\psi = \chi_R + i\chi_L$. The rotation $S(t)=R[\theta(t)]$ acts on $\psi$ as multiplication by $e^{i\theta(t)}$, and this is implemented (formally) by the operator
	\begin{equation}
		U_\theta
		=
		\exp\!\left[i\int \dd t\,\theta(t)\,\rho(t)\right],
		\qquad
		\rho(t) = :\psi^\dagger(t)\psi(t):.
	\end{equation}
	With the Fourier convention $\rho(\omega) = \int\dd t\,e^{i\omega t}\rho(t)$, the chiral density satisfies the $U(1)$ Kac--Moody algebra
	\begin{equation}
		[\rho(\omega), \rho(\omega')]
		=
		\omega\,\delta(\omega+\omega').
		\label{eq:KM}
	\end{equation}
	On the Fock vacuum, the two-point function is obtained from the commutator by frequency ordering,
	\begin{equation}
		\langle 0|\rho(\omega)\rho(\omega')|0\rangle
		=
		\omega\,\Theta(\omega)\,\Theta(-\omega')\,\delta(\omega+\omega').
		\label{eq:rho2pt}
	\end{equation}
	
	Since $U_\theta$ is Gaussian in $\rho$, the overlap $\braket{0|U_\theta|0}$ is the exponential of a quadratic form,
	\begin{equation}
		\braket{0|U_\theta|0}
		=
		\exp\!\left[-\frac{1}{2}\langle X^2\rangle_c\right],
		\qquad
		X = \int \dd t\,\theta(t)\,\rho(t) = \int \frac{\dd\omega}{2\pi}\,\widetilde{\theta}(-\omega)\,\rho(\omega),
	\end{equation}
	where $\widetilde{\theta}(\omega)=\int\dd t\,e^{i\omega t}\theta(t)$. Using \eqref{eq:rho2pt},
	\begin{align}
		\langle X^2\rangle_c
		&=
		\int\!\int \frac{\dd\omega}{2\pi}\frac{\dd\omega'}{2\pi}\,
		\widetilde{\theta}(-\omega)\widetilde{\theta}(-\omega')\,
		\omega\,\Theta(\omega)\Theta(-\omega')\,\delta(\omega+\omega')\\
		&=
		\frac{1}{4\pi^2}\int_0^\infty \dd\omega\,\omega\,|\widetilde{\theta}(\omega)|^2.
	\end{align}
	Therefore
	\begin{equation}
		-\log\abs{\braket{0|U_\theta|0}}^2
		=
		\frac{1}{4\pi^2}\int_0^\infty \dd\omega\,\omega\,|\widetilde{\theta}(\omega)|^2
		=
		\frac{1}{4\pi^2}\int_0^\infty \frac{\dd\omega}{\omega}\,|\widetilde{\dot\theta}(\omega)|^2,
	\end{equation}
	where we used $\widetilde{\dot\theta}(\omega) = i\omega\,\widetilde{\theta}(\omega)$. The integrand is manifestly non-negative. The infrared cutoff $\mu$ renders the integral finite when $\Delta\theta\neq0$.
	
	For a smooth profile with $\widetilde{\dot\theta}(0)=\Delta\theta$,
	\begin{equation}
		-\log\abs{\braket{0|U_\theta|0}}^2
		\sim
		\frac{(\Delta\theta)^2}{4\pi^2}\log\frac{\omega_*}{\mu}
		=
		\left(\frac{\Delta\theta}{2\pi}\right)^2\log\frac{\omega_*}{\mu},
	\end{equation}
	in agreement with the finite time-window result for $|\Delta\theta|<\pi$ and extending it to all $\Delta\theta$ for smooth $\theta(t)$.
	
	\paragraph{Application to sharp edges.}
	The Gaussian formula \eqref{eq:gaussian-overlap} is derived under the assumption that $\theta(t)$ is a smooth function with $\dot\theta\in L^1(\bbR)\cap L^2(\bbR)$. For an unregularized step, $\dot\theta$ is a delta function and the integral \eqref{eq:gaussian-overlap} is not defined; the formula requires the \emph{smooth} profile obtained after regularization. On a lattice, the discrete sampling of the edge is itself a regularization, and the sampled profile must traverse the arc $e^{i\theta}$ in a way that resolves the lift $\Delta\theta$. The empirical criterion of Sec.~\ref{sec:edge-sharpness} states that this happens when the sampled edge has at least one intermediate point per unit of the traversed lift. Sharp edges without such intermediate points give the folded value, which is a lattice artifact and not the physical continuum result.
	
	The formula \eqref{eq:gaussian-overlap} is a special case of the coherent-state (leviton) overlap computed in \cite{LevitovLeeLesovik1996,IvanovLeeLevitov1997,KeelingKlichLevitov2006}.

\end{document}